\documentclass[11pt]{article}
\usepackage{amsfonts}
\usepackage{amsthm}
\usepackage{amsmath}
\usepackage{amssymb}
\usepackage{latexsym}
\usepackage{bbm}
\usepackage[pdftex]{graphicx,color}
\usepackage[table,xcdraw]{xcolor}
\usepackage{enumerate}
\usepackage{verbatim}
\usepackage{amsfonts}
\usepackage{epsfig}
\usepackage{tikz}
\usepackage{hyperref}
\usepackage{bbm, color} 
\usepackage{amssymb,amsmath,tikz,graphicx,float,xspace,chemarrow}
\usepackage{graphicx}
\usepackage[utf8]{inputenc}
\makeatletter

\newcommand{\be}{\begin{equation}}
\newcommand{\ee}{\end{equation}}
\newcommand{\bea}{\begin{eqnarray}}
\newcommand{\eea}{\end{eqnarray}}

\newcommand{\distas}[1]{\mathbin{\overset{#1}{\kern\z@\sim}}}%
\newsavebox{\mybox}\newsavebox{\mysim}
\newcommand{\distras}[1]{%
  \savebox{\mybox}{\hbox{\kern3pt$\scriptstyle#1$\kern3pt}}%
  \savebox{\mysim}{\hbox{$\sim$}}%
  \mathbin{\overset{#1}{\kern\z@\resizebox{\wd\mybox}{\ht\mysim}{$\sim$}}}%
}
\newcommand{\norm}[1]{
\|{#1}\|^2}

\newtheorem{theorem}{Theorem}
\newtheorem{lemma}{Lemma}

\newtheorem{definition}{Definition}

\definecolor{LightCyan}{rgb}{0.88,1,1}
\definecolor{shadecolor}{rgb}{0.01,0.4,.8}
\begin{document}
\begin{center} { \large \sc Bayesian predictive densities as an interpretation of a class of Skew--Student $t$ distributions with application to medical data}
\vskip 0.1in {\bf Abdolnasser Sadeghkhani}
\\
\vskip 0.1in
Department of Mathematics and Statistics, Queen's University, Kingston, ON, Canada
\\ E-mail: a.sadeghkhani@queensu.ca \\

\end{center}

\vspace*{0.5cm}

\normalsize
\begin{abstract}
This paper describes a new Bayesian interpretation of a class of skew--Student $t$ distributions.
We consider a hierarchical normal model with unknown covariance matrix and show that by imposing different restrictions on the parameter space, corresponding Bayes predictive density estimators under Kullback-Leibler loss function embrace some well-known skew--Student $t$ distributions. We show that obtained estimators perform better in terms of frequentist risk function over regular Bayes predictive density estimators. We apply our proposed methods to estimate future densities of medical data: the leg-length discrepancy and effect of exercise on the age at which a child starts to walk.
\end{abstract}


\section{Introduction}
Normal density has been used to analyze data in many applications for decades. However, Roberts (1998) as an example of a weighted normal model introduced the skew--normal density for asymmetric data set, but he did not use the term ``skew--normal" then. Azzalini (1985) formalized the skew normal distribution as a generalization of a normal distribution which can be used to model asymmetric data. His work inspired many statisticians to study different versions of the skew--normal distributions rapidly. 
However, for many applications, the skew-normal model fails to provide a good fit for skewed data due to the excessive kurtosis of the data such as long--tailed data which occurs in several cases such as in finance and insurance business.
Hansen (1994) proposed a skew extension to the Student $t$ distribution for modeling financial returns. Since then, several other papers have studied different versions of skew Student $t$ distributions for financial and other applications, see e.g. Bauwens and Laurent (2005); Branco and Dey (2001); Fernandez and Steel (1998) and Patton (2004).

Following Azzalini and Valle (1996),
a random variable $T$ is said to have a  skew--normal distribution $\mathbb{SN}_p(\lambda, \xi, \tau)$, if the probability density function (pdf) is
\begin{align}\label{sn}
\frac{2}{\tau^p}\,\phi_p(\frac{t-\xi}{\tau})\,\Phi(\lambda^T \frac{t-\xi}{\tau}), \,\,\, \, \,t \in \mathbb{R}^p,
\end{align}
where $\phi_p(\cdot)$ is the pdf of a $p$ variate normal distribution, $\Phi(\cdot)$ is the cumulative density function (cdf) of a univariate standard normal density and $\lambda \in \mathbb{R}^p$ is the shape parameter which determines the skewness and $\lambda^T$ is its transpose.
Another version of density (\ref{sn}) introduced by Gupta et al. (2004) which denoted by $\mathbb{SN}_p(\alpha_0, \alpha_1, \xi, \tau)$ is given by
\begin{align}\label{sn'}
\frac{1}{\tau^p}\,\phi_p(\frac{t-\xi}{\tau})\frac{\Phi_p(\alpha_0 +\alpha_1 \frac{t-\xi}{\tau};\, 0)}{\Phi_n\left(\frac{\alpha_0}{\sqrt{1+\alpha_1^T \alpha_1}}, \cdots, \frac{\alpha_0}{\sqrt{1+\alpha_1^T \alpha_1}};\, \rho=\frac{\alpha_1^T \alpha_1}{1+\alpha_1^T \alpha_1}\right)},\,\,\, \, \,t \in \mathbb{R}^p,
\end{align}
where $\Phi_p(\cdot; \,0)$ and and $\Phi_n(\cdot; \,\rho)$ are cdf's of a $\mathbb{N}_p(0, I_p)$ and $\mathbb{N}_n(0, \Lambda)$ distributions respectively and $ \Lambda=(1-\rho)\,I_n+ \rho\,{I}_n {I}^T_n$ is the covariance matrix .

Furthermore, a Student $t$ distribution with degrees of freedom $\nu > 0$, location parameter $\xi$ and scale parameter $\tau$, denoted by $\mathbb{T}_p(\nu, \xi, \tau)$ has a density on $\mathbb{R}^P$, is given by
\begin{equation}\label{tls}
    \frac{1}{\tau^p}\frac{\Gamma(\frac{\nu+p}{2})}{\Gamma(\frac{\nu}{2})(\pi \nu)^{\frac{p}{2}}}\left(1+\frac{\norm{t-\xi}}{\nu\tau^2}\right)^{-\frac{\nu+p}{2}},\,\,\, \, \,t \in \mathbb{R}^p.
\end{equation}

 
\begin{definition}[\bf Skew--Student $t$ distribution]\label{ST1}
A random variable follows a skew--Student $t$ distribution, denoted by $\mathbb{ST}_p(\nu, \alpha_0, \alpha_1, \xi, \tau)$ with location $\xi \in \mathbb{R}^p$, scale $\tau>0$, $\nu \in \mathbb{R}_+$, $\alpha_1 \in \mathbb{R}^p$ and $\alpha_0 \in \mathbb{R}$ if $Z=(T-\xi)/\tau$ has pdf 
\begin{equation}\label{STls}
     \mathbb{T}_p (z; \xi, \tau, \nu)\, \frac{
     F_{p}\left(\nu+p, \,(\alpha_0+\alpha_1^T z)\sqrt{\frac{\nu+p}{\nu+ z^T z}} \right)}{F_{p}\left(\nu,\,\frac{\alpha_0}{\sqrt{1+ \alpha_1^T \alpha_1}}\right)},
 \end{equation}
where $\mathbb{T}_p$ denotes pdf a Student $t$ distribution in (\ref{tls}) and $F_{p}(\nu, \cdot)$ is cdf of a standard $p$--variate Student $t$ distribution.
For $p=1$ and $\alpha_0=0$, density (\ref{STls}) reduces to $2\,t(z,\nu)\, F_{1}\left(\nu+1,\,(\frac{\nu+1}{\nu+z^2})^{\frac{1}{2}}\alpha_1 z \right)$, 
and setting $\alpha_0=\alpha_1=0$ reduces (\ref{STls}) to a standard Student $t$ distribution.
\end{definition}
We may conclude another extension as follows.
\begin{definition}\label{st2}
A random variable $T$ follows a skew--Student $t$ distribution, denoted by $\mathbb{ST}_p(\nu, \alpha_0, \alpha_1, \alpha_2, \xi, \tau )$ with location $\xi \in \mathbb{R}^p$, scale $\tau>0$, $\nu \in \mathbb{R}_+$, $\alpha_1 \in \mathbb{R}^p$, $\alpha_0,\, \alpha_2 \in \mathbb{R}$, if $Z=(T-\xi)/\tau \in \mathbb{R}^p$, has pdf 
\begin{equation}\label{STlsm}
     \mathbb{T}_p (z; \xi, \tau, \nu)\, \frac{
     F_{p}\left(\nu+p,\,(\alpha_0+\alpha_1^T z)\sqrt{\frac{\nu+p}{\nu+ z^T z}} \right)-F_{p}\left(\nu+p,\,(\alpha_2+\alpha_1^T z)\sqrt{\frac{\nu+p}{\nu+z^Tz}}\right)}{F_{p}\left(\nu,\,\frac{\alpha_0}{\sqrt{1+\alpha_1^T\alpha_1}}\right)-F_{p}\left(\nu, \,\frac{\alpha_2}{\sqrt{1+ \alpha_1^T\alpha_1}}\right)},
 \end{equation}
where $\mathbb{T}_p$ and $F_{p}$ are as in  Definition \ref{ST1}.
\end{definition}
The distribution (\ref{sn}), for $p=1$, was formerly obtained by O'Hagan and Leonard (1976) via defining the following model.
\begin{equation*}
X \,|\, \theta \sim \mathbb{N}(\theta, \sigma^2)\,,\,\,\,\, \theta \,|\, \mu, \tau^2 \sim \mathbb{N}(\mu, \tau^2),\,\,\, \theta \geq \mu,
\end{equation*}
which the marginal distribution $X$ follows $\mathbb{SN}(\tau/\sigma, \mu, \sqrt{\sigma^2+\tau^2})$. \\
Liseo and Loperdido (2003) introduced another skew--normal density in a multivariate case by suppose that the covariance matrices $\Sigma$ and $\Omega$ are known and we have
\begin{equation*}
X \,|\, \theta \sim \mathbb{N}_p(\theta, \Sigma)\,,\,\,\,\, \theta \,|\, \mu \sim \mathbb{N}_p(\mu, \Omega),\,\,\, C\theta+d\leq 0,
\end{equation*}
where $C$ is a  $k \times p$ full rank matrix and $d\in \mathbb{R}^k$. They showed that the marginal density of $X$ is given by
\begin{equation*}
\frac{1}{\Phi_k(C\mu+d,\, C\Omega C^T;\,0)}\,\phi_p(x, \mu+\Sigma+\Omega)\,\Phi_k(C \Delta (\Sigma^{-1}x+\Omega^{-1}\mu)+d, C \Delta C^T;\,0),
\end{equation*}
where $\Delta^{-1}=(\Sigma^{-1}+\Omega^{-1})$, and $\Phi_p(a, C; \,0)$ is cdf of a $p$--variate normal density with mean vector $a$ and covariance matrix $C$.

Here, we use a hierarchical normal model with unknown covariance matrix in the presence of prior information on parameter spaces to obtain a class of skew--Student $t$ densities. In contrast to the Liseo and Loperdido's method which is based on marginal densities, we use another approach founded on posterior predictive density estimators to construct skew--Student $t$ distributions.
In fact, previous studies ( see, e.g. Corcuera and Giummol\`e 1999), indicates that under Kullback--Leibler (KL) loss, as defined in below
\begin{align}
L_{KL}(\theta, \hat{q}(\cdot; \, x_1, x_2))=\int_{\mathbb{R}^p} q_{\theta}(y_1)\, \log \frac{q_{\theta}(y_1)}{\hat{q}(y_1; x_1, x_2)}\, dy,\label{KL}
\end{align}
the Bayes predictive density estimator can be obtained by
\begin{equation}
\nonumber
\hat{q}_{\pi_{A}}(y_1;x_1, x_2) =  \int_{\mathbb{R}^p}  \phi(\frac{y_1-\theta_1}{\sigma_Y})   \,  \pi(\theta \mid x_1, x_2) \, d\theta\,,
\end{equation}
in other words, under the KL loss, the Bayes density coincides with the posterior predictive densities.

The rest of the paper is organized as follows: In Section \ref{setup} we introduce useful lemmas and sketch the problem set--up. Section \ref{pde} involves Bayes predictive densities and their interesting representations as skew--Student $t$ distributions, as well as the KL risk performances. Section \ref{examples} contains two different medical problems and we try to illustrate our methodology in finding the predictive density estimators, while some concluding remarks will appear in Section \ref{conclude}.
\section{Problem set--up}\label{setup}
Consider the following (canonical) normal model
\begin{equation}\label{modelu}
X_i \sim \mathbb{N}_p(\theta_i, \sigma^2 I_p),\, i=1, 2,\hspace{.3cm} Y_1 \sim \mathbb{N}_p(\theta_1, \sigma^2 I_p), \mbox{and} \hspace{.3cm} S^2\sim \sigma^2\, \chi^2_k,\,\mbox{are independent},
\end{equation}
with $k\geq 2$, $\theta_1 \in \mathbb{R}^p$, $\theta_2\in \mathbb{R}^p$, $\sigma^2\in \mathbb{R}_+$, and $\theta_1-\theta_2\in A \subseteq  \mathbb{R}^p$.
We consider that $\theta=(\theta_1, \theta_2)$ and $\sigma^2$ are unknown, and the objective is to obtain a predictive density estimate for $Y_1$.

\vspace{.05cm}
It can be verified that the joint density of ($X, S^2)$ in model (\ref{modelu}), supported on $\mathbb{R}^{2p} \times \mathbb{R}_+$, is given by 
\begin{equation}\label{jointu}
 p_{\theta, \sigma^2}(x, s^2) =\frac{(s^2)^{k/2-1}}{(2\pi \sigma^2)^p}\,\frac{\exp \left\{-\frac{1}{2\sigma^2}\left(\norm{x_1-\theta_1}+\norm{x_2-\theta_2}+s^2\right)\right\}}{(2\sigma^2)^{k/2}\Gamma(k/2)}\,.
\end{equation}
We assume the prior density
\begin{equation}\label{prioru}
\pi_{A}(\theta, \sigma^2)=\frac{1}{\sigma^2}\, \mathbb{I}_{A}(\theta_1- \theta_2),
\end{equation}
on the restricted parameter space $\theta_1-\theta_2 \in A$, where $A$ is a subset of $\mathbb{R}^p$
\begin{lemma}
\label{pos}
For model (\ref{modelu}), and prior in (\ref{prioru}), we have
\begin{enumerate}[{\bf(a)}]
\item  the posterior density $\pi(\theta \,|\, x, s^2)$ is proportional to
\begin{align}    
\left(1+\frac{\norm{x_2-\theta_2}}{s^2+\norm{x_1-\theta_1}}\right)^{-(p+k/2)}\!\!\!\left(1+\frac{\norm{x_1-\theta_1}}{s^2}\right)^{-(p+k/2)},
\end{align}
\item  the marginal posterior density is given by
\begin{equation*}
\pi(\theta_1\,|\, x, s^2) \propto \mathbb{T}_p(\nu=k,\, \xi=x_1,\tau=\frac{s}{\sqrt{k}})\,\mathbb{P}(V \in A)\,,
\end{equation*}
where 
 \begin{align}
 \label{W'}
 V\sim \mathbb{T}_p\left(\nu=k+p, \xi=\theta_1-x_2, \tau=\sqrt{\frac{s^2+\norm{x_1-\theta_1}}{p+k}}\right)\,.
 \end{align}
\end{enumerate}
\end{lemma}
\begin{proof}
See Appendix A.1. 
\end{proof}
For example,
for $p=1$ and $A=[0, \infty)$, the probability as defined in Lemma \ref{pos} is equivalent to
\begin{equation*}
F_{1}\left(k+1, \,\frac{\theta_1-x_2}{\sqrt{\frac{s^2+\norm{x_1-\theta_1}}{p+k}}}\right),
\end{equation*}
where $F_{1}(\nu, \,\cdot)$ is cdf of a standard Student $t$ distribution with degrees of freedom $\nu>1$.

Next lemma was introduced by Aitchison (1975), gives posterior predictive density, with respect to $\pi_0(\theta_1, \, \sigma^2 )=\frac{1}{\sigma^2}$ and based on $(X_1, S^2)$. For the first
step, we need the definition of the scale inverse chi squared density. A random variable $X$ is said to have a scale inverse chi squared density whenever for all degrees of freedom $\nu>0$ and scale parameter $\tau>0$, the pdf is
\begin{equation}\label{sich2}
\frac{(\tau^2\,\nu/2)^{\nu/2}}{\Gamma(\nu/2)}\, \frac{\exp[\frac{-\nu \tau^2}{2x}]}{1+\nu/2}, \, \, \, \, \, x \in \mathbb{R}\,,
\end{equation}
and we denoted by $SInv-\chi^2(\nu, \, \tau)$.   
\begin{lemma}\label{Vcanon}
For model (\ref{modelu}), the posterior predictive density of $Y_1$ (Bayes predictive density estimator) associated with the non--informative prior density $\pi_0(\theta_1, \, \sigma^2 )=\frac{1}{\sigma^2}$
$$\mathbb{T}_p\left(\nu=k, \xi=x_1, \tau=\sqrt{\frac{2 s_1^2}{k}}\right).$$
 \end{lemma}
\begin{proof}
For the non--informative prior $\pi_0(\theta_1, \sigma^2)=\frac{1}{\sigma^2}$ we have
\begin{align*}
\pi(\theta_1, \sigma^2\,|\, x_1,s^2)\propto (\sigma^2)^{-\frac{p+k}{2}-1} \exp \{\frac{-t'}{2\sigma^2}\},    
\end{align*}
where $t'=\norm{x_1-\theta_1}+s_1^2$.
This gives 
$\pi(\sigma^2 \,|\, x_1)\propto (\sigma^2)^{-(k/2+1)}\exp\{-\frac{s_1^2}{2\sigma^2}\}$. It is recognized as a kernel of scale inverse chi--square $SInv-\chi^2(k, \sqrt{\frac{s_1^2}{k}})$. (equation \ref{sich2}). Therefore we have
\begin{align*}
    q(y_1 ; x_1,s_1^2)_{\pi_0}&=\int_0^{\infty}q(y_1 \,|\, x_1, \sigma^2)\pi(\sigma^2 \,|\, x_1, s_1^2)\,d\sigma^2\\
    &\propto \int_0^{\infty}(\sigma^2)^{-p/2}\exp\{-\frac{\norm{y_1-x_1}}{2 \sigma^2}\}(\sigma^2)^{-k/2-1)}\,\exp\{-\frac{s_1^2}{2\sigma^2}\}\,d\sigma^2\\
  &\propto \int_0^{\infty}(\sigma^2)^{-\frac{p+k}{2}-1}\exp\left\{-\frac{1}{2\sigma^2}\left(\norm{y_1-X_1}+\frac{s_1^2}{2}\right)\right\}\,d\sigma^2\\
   &\propto t''^{-\frac{p+k}{2}}\int_0^{\infty} z^{\frac{p+k}{2}-1}\exp\{-z\}\, dz, \, \mbox{ with\,}\, t''=\norm{y_1-x_1}+\frac{s_1^2}{2}\\
  &\propto \left(1+\frac{2k\norm{y_1-x_1}}{ k\,s_1^2}\right)^{-(\frac{p+k}{2})}.
\end{align*}
This is the kernel of $\mathbb{T}_p( \nu=k, \xi=x_1, \tau=\sqrt{\frac{2 s_1^2}{k}})$ and hence the proof.
\end{proof}
\section{Predictive densities and their representations}\label{pde}
Given the prior $\pi_{A}(\theta, \sigma^2  )$ in (\ref{prioru}), the posterior predictive density is given by 
\begin{equation*}
\hat{q}_{\pi, A}(y_1;\, x, s)\,=\int_{\mathbb{R}^p_+}\int_{\mathbb{R}^p} q_{\theta_1, \sigma^2}(y_1)\, \pi_{A}(\theta, \sigma^2  \,|\, x, s^2)\, d\theta_1\,d\sigma^2\,,
\end{equation*}
where
\begin{align}
 \pi_{A}(\theta_1, \sigma^2 \,|\, x, s^2) &\propto e^{-\frac{s^2}{2 \sigma^2}} (\sigma^2)^{-(\frac{p+k}{2}+1)} 
\phi(\frac{\theta_1-x_1}{\sigma})\!\!\!\int\limits_{\{\theta_2; \theta_1-\theta_2 \in A\}}\!\!\!\! (\sigma^2)^{-\frac{p}{2}}\,\phi(\frac{\theta_2-x_2}{\sigma})\, d\theta_2.\nonumber\\
 &\propto e^{-\frac{s^2}{2 \sigma^2}} (\sigma^2)^{-(\frac{p+k}{2}+1)}\,\mathbb{P}(W \in A),\label{w}
\end{align} 
and $W\sim \mathbb{N}_p(\theta_1-x_2,\, \sigma^2\, I_p)$. 
\begin{lemma}\label{joint}
For model (\ref{modelu}), uniform prior (\ref{prioru}) on $A=\mathbb{R}_+^p$, by setting $\eta=\frac{1}{\sigma^2}$, we have
\begin{align}
\label{1}
\theta_1 \,|\,\eta, x, s^2 &\sim \mathbb{SN}_p\left(
\alpha_0=(x_1-x_2)\sqrt{\frac{\eta}{2}}, \alpha_1=1, \xi=1, \tau=\frac{1}{\sqrt{\eta}}\right)\,,
\end{align} and,
\begin{align}\label{2}
 \eta \,|\, x, s^2&\sim \pi_{U, A}(\eta \,|\, x, s^2)=\frac{\eta^{k/2-1}e^{-s^2\eta/2}}{\Gamma{(\frac{k}{2})}(\frac{2}{s^2})^{k/2}}\, \frac{\Phi_p\left((x_1-x_2)\sqrt{\frac{\eta}{2}}; 0\right)}{F_p\left(k, \frac{x_1-x_2}{\sqrt{2 s^2/k}}\right)}\,,
\end{align}
where $\mathbb{SN}_p$ is defined in (\ref{sn'}) and $F_{p}(\nu, \,\cdot)$, is cdf of a standard $p$--variate Student $t$ distribution with degrees of freedom $\nu$.
\end{lemma}
\begin{proof}
The probability in (\ref{w}) can be replaced by $\Phi_p(\frac{\theta_1-x_2}{\sigma}; 0)$. Thus
\begin{align*}
\pi_{U, A}(\theta_1, \sigma^2 \,|\, x, s^2)&\propto
\frac{e^{\frac{-s^2}{2\sigma^2}}}{(\sigma^2)^{(k/2+1)}}(\frac{1}{\sigma^2})^{\frac{p}{2}}\phi_p(\frac{\theta_1-x_1}{\sigma}\, )\Phi_p(\frac{\theta_1-x_1}{\sigma}; 0)\\
&\propto \Phi_p(\frac{x_1-x_2}{\sqrt{2}\tau}; 0)\frac{e^{\frac{-s^2}{2\sigma^2}}}{(\sigma^2)^{(k/2+1)}} \frac{(\frac{1}{\tau})^{p}\phi_p(\frac{\theta_1-\xi}{\tau})\Phi_p(\alpha_0+\alpha_1\frac{\theta_1-\xi}{\tau}; 0)}{\Phi_p(\frac{\alpha_0}{\sqrt{1+\alpha_1^2}}; 0)},
\end{align*}
Changing the variable $\eta=1/\sigma^2$ proves (\ref{1}).\\
In addition, we can write
\begin{align}\label{here}
    \pi_{U, A}(\eta \,|\, x, s^2) \propto \frac{\eta^{k/2-1}e^{-s^2\eta/2}}{\Gamma(\frac{k}{2})(\frac{2}{s^2})^{k/2}}\, \Phi_p\left((x_1-x_2)\sqrt{\frac{\eta}{2}}; 0\right).
\end{align}
Now, using the identity provided by Azzalini and Capitanio (2003) as follows
\begin{equation}\label{iden}
\mathbb{E}[\Phi_p(c\,\sqrt{\eta}; 0)]=F_{p}\left(2a,\, c\,\sqrt{\frac{a}{b}}\right),\,
\mbox{for \,} \,\eta \sim Gamma(a, b),
\, a>0, \,b>0, \,c>0\,,
\end{equation}
and choosing $a=k/2$, $b=s^2/2$ in it, (\ref{here}) can be written as
\begin{equation*}
  \pi_{U, A}(\eta \,|\, x, s^2) = \frac{\eta^{k/2-1}e^{-s^2\eta/2}}{\Gamma(\frac{k}{2})(\frac{2}{s})^{k/2}}\, \frac{\Phi_p\left((x_1-x_2)\sqrt{\frac{\eta}{2}}; 0\right)}{F_p\left(k,\, \frac{x_1-x_2}{\sqrt{\frac{2s^2}{k}}}\right)}.  
\end{equation*}
This completes the proof of (\ref{2}).
\end{proof}

\begin{theorem}
\label{adgeneral1}
For model (\ref{modelu}), the Bayes predictive density estimator based on additional prior $\hat{q}_{\pi, A}(y_1; x, s^2)$ associated with a uniform prior (\ref{prioru}) on $A=\mathbb{R}_+^P$, is given by
\begin{equation*}
\mathbb{ST}_p\left(\nu=k,  \alpha_0=\sqrt{\frac{2}{3}}\frac{x_1-x_2}{\sqrt{2 s^2/k}}, \alpha_1=1/\sqrt{3}, \,\xi=x_1,\,\tau=\sqrt{\frac{2s^2}{k}} \right),  
\end{equation*}
where $\mathbb{ST}_p$ is skew--Student $t$ distribution, as defined in Definition \ref{ST1}. Equivalently, we can write
\begin{equation}
\label{pde2} \mathbb{T}_p\left(\nu=k, \xi=x_1, \tau=\sqrt{\frac{2s^2}{k}}\right)\frac{F_{p}\left(k+p,\,\left(\sqrt{\frac{2}{3}}(x_1-x_2)+\frac{y_1-x_1}{\sqrt{3}}\right)\sqrt{\frac{k+1}{2s^2+(y_1-x_1)^2}}\right)}{F_p(k,\,\frac{x_1-x_2}{\sqrt{2s^2/k}})} \,,
\end{equation}

\end{theorem}

\begin{proof}
See Appendix A.2.
\end{proof}
\begin{lemma}
\label{jointm}
For model (\ref{modelu}) and uniform prior (\ref{prioru}) on $A=[-m,m]^p$, the marginal posterior distribution $\pi(\theta_1 \,|\, \eta, x, s^2$) is given by
\begin{small}
 $$\mathbb{SN}_p\left(
 \alpha_0=(x_1-x_2+m)\sqrt{\frac{\eta}{2}}, \alpha_1=1, \alpha_
 2=(x_1-x_2-m)\sqrt{\frac{\eta}{2}},\xi=1, \tau=\frac{1}{\sqrt{\eta}}\right).$$
\end{small}
 and also,
 \begin{align*}
 \pi(\eta \,|\, x, s^2)=\frac{\eta^{k/2-1}e^{-s\eta/2}}{\Gamma{(\frac{k}{2})}(\frac{2}{s})^{k/2}}\, \frac{\Phi_p\left((x_1-x_2+m)\sqrt{\frac{\eta}{2}}; 0\right)-\Phi_p\left((x_1-x_2-m)\sqrt{\frac{\eta}{2}}; 0\right)}{F_p(k, \frac{x_1-x_2+m}{\sqrt{2s^2/k}})-F_p(k, \frac{x_1-x_2-m}{\sqrt{2s^2/k}})}.
\end{align*}
\end{lemma}
\begin{proof}
The proof may be easily derived from a similar analysis to Lemma (\ref{joint}) with the probability $\mathbb{P}(W \in A)=\Phi_p(\frac{\theta_1-x_2+m}{\sigma}; 0)-\Phi_p(\frac{\theta_1-x_2-m}{\sigma}; 0)$ in (\ref{w}).
\end{proof}

\begin{theorem}\label{pdem}
For model (\ref{modelu}), the Bayes predictive density estimator based on additional prior information $\hat{q}_{\pi, A}(y_1; x, s^2)$, $A=[-m,m]^p$ with $m>0$, and a uniform prior (\ref{prioru}), is given by
\begin{small}
\begin{equation}\label{pde1m}
\mathbb{ST}_p\left(\alpha_0=\sqrt{\frac{2}{3}}\frac{x_1-x_2+m}{\sqrt{2s^2/k}}, \alpha_1=\frac{1}{\sqrt{3}}, \alpha_2=\sqrt{\frac{2}{3}}\frac{x_1-x_2-m}{\sqrt{2s^2/k}} \,\xi=x_1,\,\tau=\sqrt{\frac{2s^2}{k}} \right).
\end{equation}
\end{small}
In other words,
\begin{align}\label{pde2m}  \mathbb{T}_p\left(\nu=k, \xi=x_1, \tau=\sqrt{\frac{2s^2}{k}}\right)
\frac{F_{p}\left(k+1, L_1(x, s^2)\right)-F_{p}\left(k+1, L_2(x, s^2)\right)}{F_p\left(1, \frac{x_1-x_2+m}{\sqrt{2s^2/k}}\right)-F_p\left(1, \frac{x_1-x_2-m}{\sqrt{2s^2/k}}\right)} \,,
\end{align}
where $L_1(x, s^2)=\sqrt{\frac{2}{3}}(x_1-x_2+m)+\frac{y_1-x_1}{\sqrt{3}}\sqrt{\frac{k+1}{2s^2+\norm{y_1-x_1}}}$, $L_2(x, s^2)=\sqrt{\frac{2}{3}}(x_1-x_2-m)+\frac{y_1-x_1}{\sqrt{3}}\sqrt{\frac{k+1}{2s^2+\norm{y_1-x_1}}}$.
\end{theorem}
\begin{proof}
The proof is straightforward and analogous to Theorem \ref{adgeneral1}.  
\end{proof}
\subsection{Risk performance}
It would be interesting to compare the frequentist risk performance of the Bayes predictive density estimator $\hat{q}_{\pi, A}$ (the Bayes estimator with considering additional information from Theorem \ref{adgeneral1} or \ref{pdem}, depending on $A$) and $q_{\pi_0}$ (the Bayes estimator without considering additional information from Lemma 2.2). For model (\ref{modelu}) the KL risk function is given by
\begin{align*}
R_{KL}(\theta, \hat{q}) \,=\,  \int_{\mathbb{R}^{p}}\!\int_{\mathbb{R}^{p}}  L_{KL}\left(\theta, \hat{q}(\cdot;x)\right) \,p_{\theta}(x_1, x_2) \, dx_1\,dx_2  \,.  
\end{align*}
where $p_{\theta}(x_1, x_2)=\frac{1}{\sigma_{X_1}\sigma_{X_2}}\phi(\frac{x_1-\theta_1}{\sigma_{X_1}})\phi(\frac{x_2-\theta_2}{\sigma_{X_2}})$.

Figure \ref{U0}, and \ref{U1} present the relative efficiency (risk ratio) of $\hat{q}_{\pi, A}$ over $\hat{q}_{\pi_0}$, for $p=1$, $k=3$ and $\Delta=(\theta_1-\theta_2)/\sigma$ based on restricted parameter $A=[0, \infty)$ and $A=[-6, 6]$ respectively. For both graphs, we have about $12\%$ improvement in the KL risk function.
\begin{figure}[H] 
    \centering 
    \includegraphics[width=0.5\textwidth]{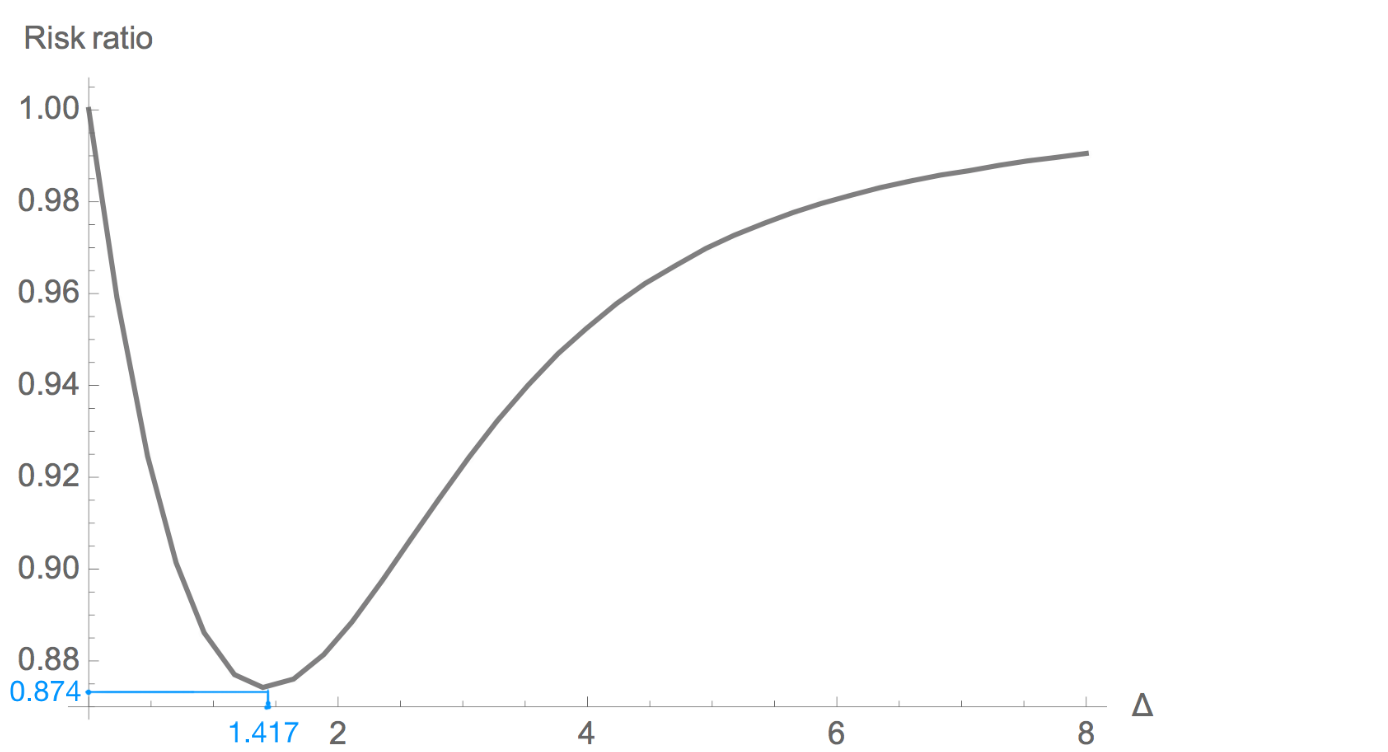}
    \caption{Risk ratio of the Bayes predictive density estimators for $p=1$, $k=3$ and $A=[0, +\infty)$.  }\label{U0}
\end{figure}
\begin{figure}[H] 
    \centering 
    \includegraphics[width=0.55\textwidth]{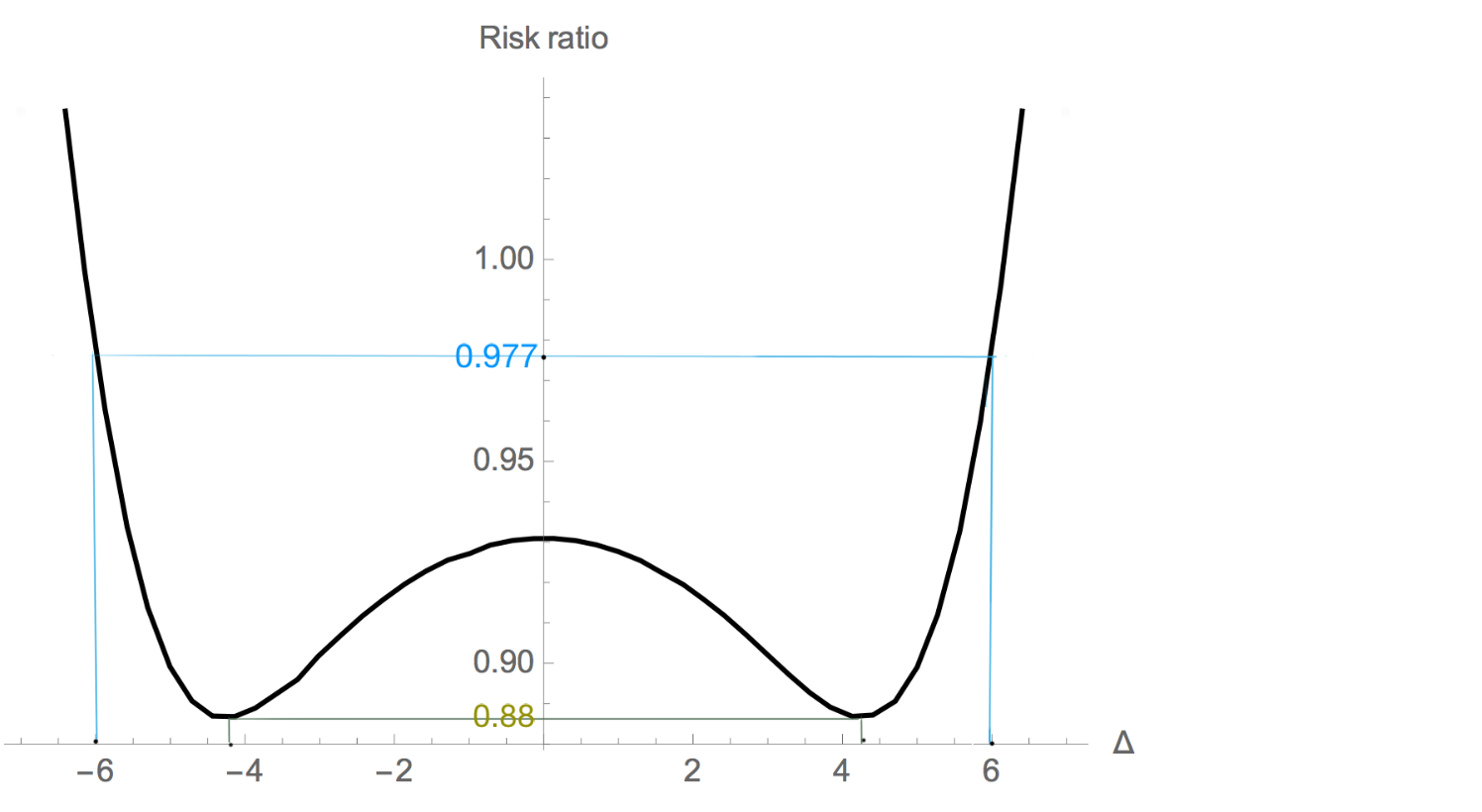}
    \caption{Risk ratio of the Bayes predictive density estimators for $p=1$, $k=3$ and $A=[-6, 6]$.  }\label{U1}
\end{figure}
\section{Examples}\label{examples}
In this section, we apply the proposed methods in order to construct Bayes predictive density estimators through two well-known medical data.\\

\textbf{\large{ Example 1 (Leg-length discrepancy predictive density estimation)}}

Leg--length Discrepancy (LLD), the difference between the lengths of two legs, is a topic that seemingly has been exhaustively examined. The LLD may be caused by trauma or mild developmental abnormalities, with onset in birth or childhood. In fact, it causes several conditions, including low back pain; osteoarthritis of the hip and knee; knee pain and running injuries, such as Achilles rupture.
Harvey et al. (2010) used radiography to evaluate leg length in 3,026 adults. After following participants for 30 months they conducted exploratory analyses to determine whether there was an important threshold value of the LLD above which knee osteoarthritis was more likely. They did this by stratifying the LLD into four categories: less than 0.5 cm (reference group), 0.5 cm to less than 1 cm, 1 cm to less than 2 cm, and 2 cm or more. Their result showed leg-length inequality of 1 cm or more to be associated with prevalent, incident, symptomatic and progressive knee osteoarthritis that was strongest in the shorter leg. \\
Table \ref{Baseline} shows the body mass of participants grouped with the LLD (defined as inequality of 1 cm or more).

\begin{table}[H]
\centering
\label{Baseline}
\begin{tabular}{|l|lll|}
\hline
Body mass index  &sample size \cellcolor[HTML]{C0C0C0} & mean\cellcolor[HTML]{C0C0C0}& sd  \cellcolor[HTML]{C0C0C0} \\  \hline LLD greater or equal than 1 cm (group 1)
\cellcolor[HTML]{C0C0C0} &429 &31 & 5.7    \\
\hline LLD less than 1 cm (group 2)
\cellcolor[HTML]{C0C0C0} & $2535$  & $30.4$& 5.7                          \\\hline                    
 \end{tabular}
\caption{Patient body mass}
\end{table}
Also it is statistically significant at the $0.05$ level the mean of body mass index in the group with LDD $\geq 1$ is greater than group with LDD $< 1$. Suppose random variable $X_1$, the body mass index with LDD $\geq1$, follows $\mathbb{N}(\theta_1, \sigma_1^2)$, is independent of $X_2$, the body mass index with LDD $<1$ which is distributed as $\mathbb{N}(\theta_2, \sigma_2^2)$, when their means are subject to the order restriction $\theta_1 \geq \theta_2$ and variances $\sigma_1^2$ and $\sigma_2^2$ are unknown. 

Table 2 contains the predictive density estimators $\hat{q}_{\pi, A}$, and the Bayes estimators $q_{\pi_0}$ (
predictive density estimators without and with considering the additional information respectively) along with their means, 10$^{th}$, 50$^{th}$ and 90$^{th}$ percentiles for the future density $Y_1$ of the body mass index of patients with LLD$\geq$ 1 cm, 
based on the data from Table 1.
 \begin{table}[H]
\centering
\resizebox{\columnwidth}{!}{
\label{PDEs1212}
\begin{tabular}{|l|l|l|l|}
\hline
\rowcolor[HTML]{C0C0C0} 
\small{Predictive Density estimation}& \small{Estimator} &\small{PDF}  &   $\bar{y}_1$, $P_{0.1}$, $P_{0.5}$, $P_{0.9}$\\ \hline
\small{Bayes without additional information}& $\mathbb{T}_1(n_1-1, \bar{x}_1, \sqrt{\frac{2s_1^2}{n_1-1}})$& $\mathbb{T}_1(428, 31, 0.39)$ & 31, 30.5, 31, 31.5 \\\hline
\small{Bayes with additional information} & $\mathbb{ST}_1(n_1-1, \frac{\bar{x}_1-\bar{x}_2}{\sqrt{3s_1^2/(n_1-1)}},\bar{x}_1,\sqrt{\frac{2s_1^2}{n_1-1}})$& $\mathbb{ST}_1(428, 1.26, 31,0.39)$ & 31.02, 30.52, 31.02, 31.52 \\\hline
\end{tabular}
}
\caption{Predictive density estimators of future density of the body mass index of patients with LLD$\geq$ 1 cm, along with their means, 10$^{th}$, 50$^{th}$ and 90$^{th}$ percentiles.}
\end{table}

\textbf{\large{ Example 2 (Child's first walk)}}

An experiment was conducted to evaluate the effect of exercise on the age at which a child starts to walk (see Silvapulle and Sen, 2005). Let $X$ denote the age (in months) at which a child starts to walk.
\begin{table}[H]
\centering
\label{TableWalking}
\begin{tabular}{|l|cccccc|l|}
\hline
Group 1\cellcolor[HTML]{C0C0C0} &  11& 10& 10& 11.75& 10.5& 15&  $\bar{x}_1=11.37$, $s_1=1.44$ \\ \hline
Group 2\cellcolor[HTML]{C0C0C0} &  9 & 9.5& 9.75 &10 &13 &9.5  & $\bar{x}_2=10.12$, $s_2=1.9$   \\ \hline
\end{tabular}
\caption{The age at which a child first walks}
\end{table}
The first group performed daily exercises but not the special walking exercises while the second group performed a special walking exercise for 12 minutes per day beginning at age 1 week and lasting 7 weeks. (the original experiment consists of other groups, however, here we consider only two of them.)\\
For groups $i (i = 1, 2)$ let, $\theta_i$ be the mean age (in months) at which a child starts to walk. However, suppose that the researcher was prepared to assume that the walking exercises would not have negative effect of increasing the mean age at which a child starts to walk, and it was desired that this additional information be incorporated to improve on the statistical analysis. In this case, we have that $\theta_1 \geq \theta_2$. 
this can be considered as two univariate normal distributions $X_1$, $X_2$, when their means are subject to the order restriction $\theta_1 \geq \theta_2$ and variances $\sigma_1^2$ and $\sigma_2^2$ are different and unknown. Analogous to example 1, Table 4 can be similarly obtained. 
 \begin{table}[H]
\centering
\resizebox{\columnwidth}{!}{
\label{PDEs11}
\begin{tabular}{|l|l|l|l|}
\hline
\rowcolor[HTML]{C0C0C0} 
\small{Predictive Density estimation}& \small{Estimator} &\small{PDF}  &   $\bar{y}_1$, $P_{0.1}$, $P_{0.5}$, $P_{0.9}$\\ \hline
\small{Bayes without additional information}& $\mathbb{T}_1(n_1-1, \bar{x}_1, \sqrt{\frac{2s_1^2}{n_1-1}})$& $\mathbb{T}_1(5, 11.37, 1.2)$ & 11.37, 9.6, 11.37, 13.14 \\\hline
\small{Bayes with additional information} & $\mathbb{ST}_1(n_1-1, \frac{\bar{x}_1-\bar{x}_2}{\sqrt{3s_1^2/(n_1-1)}},\bar{x}_1,\sqrt{\frac{2s_1^2}{n_1-1}})$& $\mathbb{ST}_1(5, 0.85, 11.37,1.2)$ & 11.45, 11.2, 11.44, 12.37 \\\hline
\end{tabular}
}
\caption{Predictive density estimators of future density of child first walks in group 1 along with their means, 10$^{th}$, 50$^{th}$ and 90$^{th}$ percentiles.}
\end{table}
Figure 3, helps to visualize different predictive density estimators and the corresponding means and percentiles in Table 4.
\begin{figure}[H] \label{vis}
    \centering 
    \includegraphics[width=0.85\textwidth]{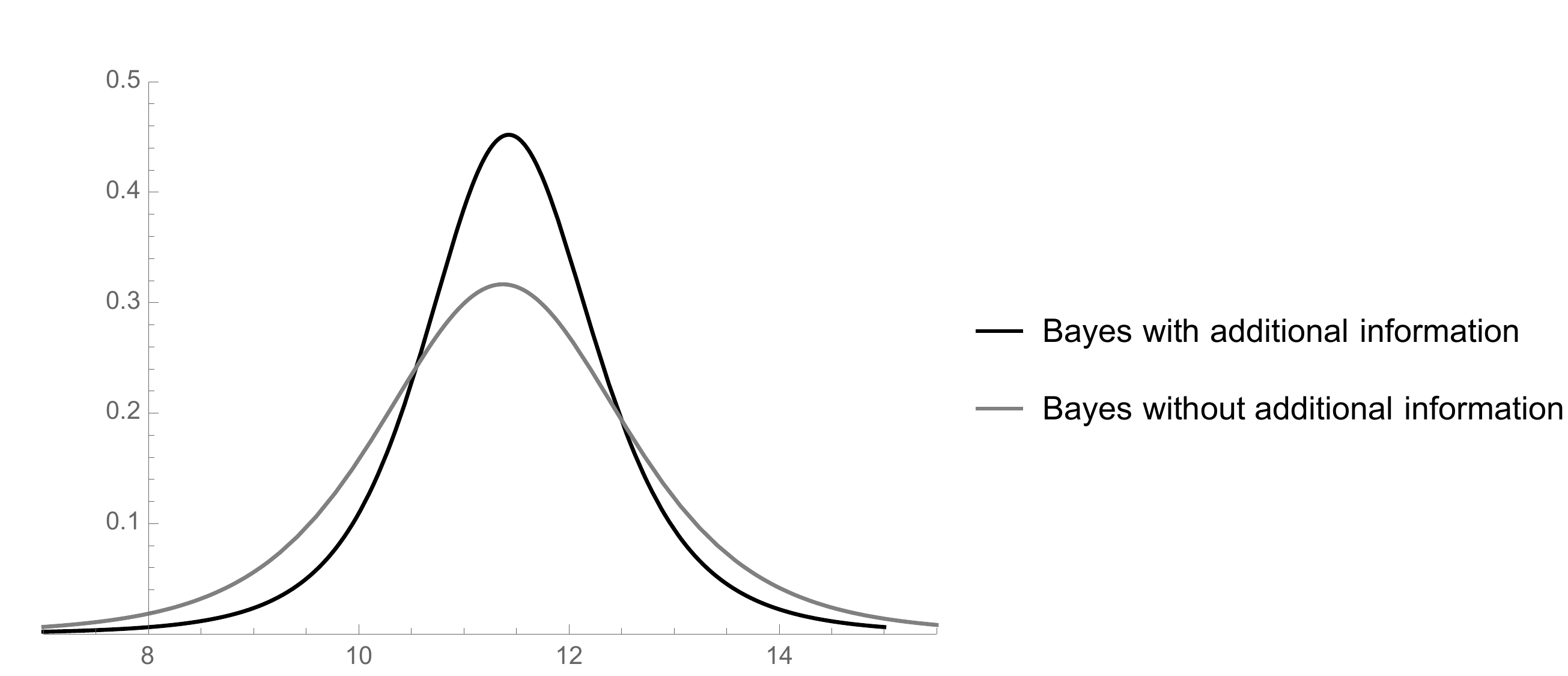}
    \caption{Visualization of Table 4. }
\end{figure}

\section{Concluding remarks}\label{conclude}
This paper extends the line of work which seeks to find Bayesian interpretations of the skew--normal densities to skew--Student $t$ distributions. We have shown that different kind of constraints on the parameter space in a hierarchical normal model, yield the Bayesian predictive densities belong to a class of weighted Student $t$ distributions. More specifically we studied the restrictions $\theta_1-\theta_2 \in \mathbb{R}_+^p$, and $\theta_1-\theta_2 \in [-m, m]^p$, in model (\ref{modelu}), which provides two different skew--Student $t$ distributions based on Definitions \ref{ST1} and \ref{st2} respectively. Results suggest Bayes predictive density estimators based on additional information performs better than the Bayes predictive density without considering additional information in term of KL risk function. Finally, some numerical comparison and important examples were done to support the results. 
\section*{Acknowledgement}
The author thanks Éric Marchand (Université de Sherbrooke) for his very helpful comments on the manuscript.
\section*{Appendix}\label{proof}
{\bf (A.1) Proof of Lemma \ref{pos}}\\

\begin{proof}
\textbf{(a)} We have
\begin{equation}\label{ttt}
\pi(\theta, \sigma^2\,|\, x, s^2)\propto (\sigma^2)^{-(p+k/2+1)} \exp \left\{\frac{-t}{2\sigma^2}\right\}\,\mathbb{I}_{A}(\theta_1- \theta_2)\,,
\end{equation}
where $ t=\norm{x_1-\theta_1}+\norm{x_2-\theta_2}+s^2$.
Now, letting $z=\frac{t}{2\sigma^2}$,  by integrating out $\sigma^2$, we have
\begin{align*}
\pi(\theta \,|\, x, s^2)&\propto\int_0^{\infty} (\sigma^2)^{-(p+k/2+1)}\exp\{-\frac{t}{2\sigma^2}\}\,\mathbb{I}_{A}(\theta_1- \theta_2)\, d\sigma^2\\
&\propto t^{-(p+k/2)}\,\mathbb{I}_{A}(\theta_1- \theta_2)\int_0^{\infty} z^{p+k/2-1}\exp\{-z\} \,dz\\  
&\propto
t^{-(p+k/2)}\,\mathbb{I}_{A}(\theta_1- \theta_2)\\
&\propto(\norm{x_1-\theta_1}+\norm{x_2-\theta_2}+s^2)^{-(p+k/2)}\mathbb{I}_{A}(\theta_1- \theta_2)\\
&\propto\left(1+\frac{\norm{x_2-\theta_2}}{s^2+\norm{x_1-\theta_1}}\right)^{-(p+k/2)}\!\!\left(1+\frac{\norm{x_1-\theta_1}}{s^2}\right)^{-(p+k/2)}\!\!\mathbb{I}_{A}(\theta_1- \theta_2).
\end{align*}
\textbf{(b)} We have
\begin{align*}
\pi(\theta_1 &\,|\, x, s^2)=\!\!\!\!\!\!\int\limits_{\{\theta_2:\,\theta_1-\theta_2 \in A\}}\!\!\!\!\!\! \pi(\theta \,|\, x, s^2)\, d\theta_2\\
&\propto  \left(1+\frac{\norm{x_1-\theta_1}}{s^2}\right)^{-(p+k/2)} 
\!\!\!\!\!\!\!\!\int\limits_{\{\theta_2: \,\theta_1-\theta_2 \in A\}} \!\!\!\!\!\! \left(1+\frac{\norm{x_2-\theta_2}}{s^2+\norm{x_1-\theta_1}}\right)^{-(p+k/2)}\,d\theta_2,\\
&\propto  \left(1+\frac{\norm{\theta_1-x_1}}{k(s^2/k)}\right)^{-(p+k/2)} 
\!\!\int\limits_{A}  \left(1+\frac{\norm{t-(\theta_1-x_2)}}{s^2+\norm{x_1-\theta_1}}\right)^{-(p+k/2)}\,dt\,,
\end{align*}
by the change of variable $t=\theta_1-\theta_2$.
\end{proof}
{\bf (A.2) Proof of Theorem \ref{adgeneral1}}.\\

By setting $U=\eta(\theta_1-x_1)$, and $\eta=\frac{1}{\sigma^2}$, one can write the joint density of $(U, \eta)\,|\, x, s^2$ as multiplication of equations (\ref{1}) and (\ref{2}) in Lemma \ref{joint}. So we have
\begin{align}
 \hat{q}_{\pi, A}&(y_1; x, s^2)\,=\,\mathbb{E}^{\,(U, \eta) \,|\, x, s^2} q(y_1 \,|\, x_1+\frac{U}{\sqrt{\eta}},\, \frac{1}{\eta}) \nonumber\\
 &=\int_{\mathbb{R}_+^p} \left(\int_{\mathbb{R}^p}(\frac{\eta}{2 \pi})^{\frac{p}{2}} e^{-\frac{\eta}{2}\norm{y_1-x_1-\frac{u}{\sqrt{\eta}}}}\frac{\phi_p(u)\Phi_p(\alpha_0+\alpha_1 u; 0)}{\Phi_p(\frac{\alpha_0}{\sqrt{1+\alpha_1^T\alpha_1}}; 0)}du\right)\pi(\eta \,|\, x, s^2)\,d\eta \nonumber\\
  &=\int_{\mathbb{R}_+^p} \frac{\eta^{\frac{k+p}{2}-1}e^{-\eta/2(s^2+\norm{y_1-x_1})}}{F_p(k; \frac{x_1-x_2}{\sqrt{2s^2}/k})\Gamma{(\frac{k}{2})}(\frac{2}{s^2})^{k/2}}\int_{\mathbb{R}^p}\frac{e^{-\norm{u}/2+\sqrt{\eta} u^T(y_1-x_1)}}{(2\pi)^\frac{p}{2}}\phi_p(u)\Phi_p(\alpha_0+\alpha_1 u; 0)\,du\,d\eta  \nonumber\\
 &=\frac{\Gamma(\frac{k+p}{2})}{\Gamma(\frac{k}{2})(\frac{2}{s^2})^{\frac{k}{2}}(2\pi)^{\frac{p}{2}}F_p(k, \frac{x_1-x_2}{\sqrt{2s^2}/k})}\left(\frac{s^2}{2}+\frac{\norm{y_1-x_1}}{4}\right)^{-\frac{k+p}{2}}
\nonumber \\
 & \times \int_{\mathbb{R}_+^p}\eta^{\frac{k+1}{2}-1} e^{-\eta \left(\frac{s^2}{2}+\frac{\norm{y_1-x_1}}{4}\right)}\frac{1}{\sqrt{2}}\Phi_p\left(\sqrt{\eta}\left(\frac{x_1-x_2}{\sqrt{3}}+\frac{y_1-x_1}{\sqrt{6}}; 0\right)\right)\, d\eta\nonumber \\
 &= \frac{\Gamma(\frac{k+p}{2})}{\Gamma(\frac{k}{2})(\frac{2}{s^2})^{\frac{k}{2}}(2\pi)^{\frac{p}{2}}F_p(k; \frac{x_1-x_2}{\sqrt{2s^2}/k})}\frac{1}{\sqrt{2}}\mathbb{E}^{\eta \,|\,x, s^2} \left[\Phi_p\left(\sqrt{\eta}\left(\frac{x_1-x_2}{\sqrt{3}}+\frac{y_1-x_1}{\sqrt{6}}\right); 0\right)\right],\nonumber
\end{align}
where $\eta \,|\,x, s^2 \sim Gamma(\frac{k+1}{2}, \, \frac{s^2}{2}+\frac{\norm{y_1-x_1}}{4}),$ applying the identity  (\ref{iden}) to above expectation completes the proof.

\end{document}